\documentclass[graybox]{svmult}

\usepackage{mathptmx}       
\usepackage{helvet}         
\usepackage{courier}        
\usepackage{type1cm}        
\usepackage{makeidx}         
\usepackage{graphicx}        
\usepackage{multicol}        
\usepackage[bottom]{footmisc}

\makeindex             

\begin{document}

\title*{Ground states for Potts model with a countable set of spin values on a Cayley tree}
\author{G.I.Botirov and M.M.Rahmatullaev}
\institute{G.I.Botirov \at Institute of Mathematics,
Tashkent, Uzbekistan, \email{botirovg@yandex.ru}
\and M.M.Rahmatullaev \at Institute of Mathematics,
Tashkent, Uzbekistan, \email{mrahmatullaev@rambler.ru}}
%
%
\maketitle

\abstract*{We consider Potts model, with competing interactions and countable spin values $\Phi=\{0,1,\dots \}$ on a Cayley tree of order three. We study periodic ground states for this model.}

\abstract{We consider Potts model, with competing interactions and countable spin values $\Phi=\{0,1,\dots \}$ on a Cayley tree of order three. We study periodic ground states for this model.}

\section{Introduction}
\label{sec:1}
Each Gibbs measure is associated with a single phase of physical system. As is known, the phase diagram of Gibbs measures for a Hamiltonian is close to the phase diagram of isolated (stable) ground states of this Hamiltonian. At low temperatures, a periodic ground state corresponds to a periodic Gibbs measures, see \cite{rm}-\cite{rm2013}. The problem naturally aries on description of periodic ground states. In \cite{r2006} and \cite{rm2010} for the Ising model with competing interactions, periodic and weakly periodic ground states were studied.

In \cite{rm} ground states were described and the Peierls condition for the Potts model is verified. Using a contour argument authors showed the existence of three different Gibbs measures associated with translation invariant ground states.

In \cite{br}, (\ref{H}), \cite{rras} studying periodic and weakly periodic ground states for the Potts model with competing interactions on a Cayley tree.
In the present paper, we consider Potts model, with competing interactions and a countable set of spin values $\Phi=\{0,1,\dots \}$ on a Cayley tree of order three. We study periodic ground states.

The paper is organized as follows. In Section 2, we recall the main definitions and known facts. In Section 3, we study periodic ground states.

\section{Main definitions and known facts}
\textbf{Cayley tree.} The Cayley tree (Bethe lattice) $\Gamma^k$ of order $k \geq 1$ is an infinite tree, i.e., a graph
without cycles, such that exactly $k+1$ edges originate from each vertex (see \cite{B}). Let $\Gamma^k=(V,L)$ where
$V$ is the set of vertices and $L$ the set of edges. Two vertices $x$ and $y$ are called {\it nearest neighbors}
if there exists an edge $l \in L$ connecting them and we denote $l=\langle x,y\rangle $.

On this tree, there is a natural distance to be denoted $d(x; y)$, being the number of
nearest neighbor pairs of the minimal path between the vertices $x$ and $y$ (by path one
means a collection of nearest neighbor pairs, two consecutive pairs sharing at least a
given vertex).

For a fixed $x^0 \in V$, the root, let
$$W_n=\{x \in V: d(x,x^0)=n\}, \ \ \ \ V_n=\{x \in V : d(x,x^0) \leq n\};$$
be respectively the sphere and the ball of radius $n$ with center at $x^0$, and for $x \in W_n$ let
$$S(x)=\{y \in W_{n+1}: d(x,y)=1\}$$
be the set of direct successors of $x$

It is well-known that there exists a one-to-one correspondence between the set $V$ of vertices
of the Cayley tree of order $k \geq 1$ and the group $G_k$ of the free products of $k + 1$ cyclic
groups of second order with generators $a_1, a_2, \dots , a_{k+1}$ (see \cite{g}, \cite{r}).

\section{Configuration space and the model}
\label{sec:2}
For each $x\in G_k$, let $S(x)$ denote the set of direct successors of $x$, i.e., if $x\in W_n$ then
$$
S(x)=\{y\in W_{n+1}:d(x,y)=1\}.
$$

For each $x\in G_k$, let $S_1(x)$ denote the set of all neighbors of $x$, i.e. $S_1(x)=\{y\in G_k: \langle x,y\rangle\in L\}.$ The set $S_1(x)\setminus S(x)$ is a singleton. Let $x_\downarrow$ denote the (unique) element of this set.

We consider the models in which the spin takes values in the set $\Phi=\{1,2, \dots\}$. A {\it configuration} $\sigma$ on the set $V$ is then defined as a function $x \in V \rightarrow \sigma(x) \in \Phi$; the set of all configurations coincides with $\Omega=\Phi^V$.

Let $G_k^*$ be a subgroup of index $r \geq 1$. Consider the set of right coset $G_k / G_k^*=\{H_1, \dots, H_r\}$, where $G_k^*$ is a subgroup.

\begin{definition} A configuration $\sigma(x)$ is said to be $G_k^*$ - {\it periodic } if $\sigma(x)=\sigma_i$ for all $x \in H_i$. A $G_k$-periodic configuration is said to be {\it translation invariant}.
\end{definition}

The period of a periodic configuration is the index of the corresponding subgroup.

\begin{definition} A configuration $\sigma(x)$ is said to be $G_k^*$ - {\it weakly periodic } if $\sigma(x)=\sigma_{ij}$ for all $x \in H_i$ and $x_\downarrow\in H_j$.
\end{definition}
The Hamiltonian of the Potts model with competing interactions has the form

\begin{equation}\label{H} H(\sigma)=J_1 \sum \limits_{\langle x,y\rangle: \atop x, y \in V} \delta_{\sigma(x)\sigma(y)} +J_2  \sum \limits_{x,y \in V: \atop d(x, y)=2} \delta_{\sigma(x)\sigma(y)},
\end{equation}
where $J_1, J_2\in \mathbf{{R}}$ and
$$\delta_{u v}=\left\{\begin{array}{ll} 1, \ \ \textrm{если} \ \
 u=v,\\
0,\ \ \textrm{если} \ \ u\ne v.\\
\end{array}\right.
$$

\section{Ground states}
For pair of configurations $\sigma$ and $\varphi$ coinciding almost everywhere, i.e., everywhere except at a finite number of points, we consider the relative Hamiltonian $H(\sigma, \varphi)$ of the difference between the energies of the configurations $\sigma$ and $\varphi$, i.e.,
\begin{equation}\label{3.1}
H (\sigma, \varphi) =J_1\sum\limits _ {\langle x, y
\rangle, \atop x, y \in V} (\delta _ {\sigma (x) \sigma
(y)}-\delta _ {\varphi (x) \varphi (y)}) +J_2\sum\limits _ {x, y
\in V: \atop d (x, y) =2} (\delta _ {\sigma (x) \sigma (y)}
-\delta _ {\varphi (x) \varphi (y)}),
\end{equation}
where $J=(J_1, J_2)\in\mathbf{R}^2$ is an arbitrary fixed parameter.

Let $M$ be the set of all unit balls with vertices in $V$. By the {\it restricted configuration} $\sigma_b$ we mean the restriction of a configuration $\sigma$ to a ball $b\in M$. The energy of a configuration $\sigma_b$ on $b$ is defined by the formula
\begin{equation}
U (\sigma_b) \equiv U (\sigma_b, J) = \frac {1} {2} J_1\sum\limits _ {\langle x, y
\rangle, \atop x, y \in b} \delta _ {\sigma (x) \sigma (y)}
+J_2\sum\limits _ {x, y \in b:\atop d (x, y) =2} \delta _ {\sigma
(x) \sigma (y)}, \end{equation} where $J=(J_1, J_2)\in\mathbf{R}^2$.

The following assertion is known (see \cite{r}-\cite{rras})
\begin{lemma}
The relative Hamiltonian (\ref{3.1}) has the form
$$H (\sigma, \varphi) = \sum\limits _
{b\in M} (U (\sigma_b)-U (\varphi_b)) .$$
\end{lemma}

Note that, in \cite{br} in the case $k=2$ and $\Phi=\{1,2,3\}$ all periodic (in particular translation-invariant) ground states for the Potts model (\ref{H}) are given. In \cite{rm2013} the set of weakly periodic ground states corresponding to index-two normal divisors of the group representation of the Cayley tree is given. In \cite{rras} the sets of periodic and weakly periodic ground states corresponding to normal subgroups of the group representation of the Cayley tree of index 4 are described.

We consider the case $k=3$. It is easy to see that $U(\sigma_b)\in \{U_1,U_2,..., U_{12}\}$ for any $\sigma_b$, where
$$ U_1=2J_1+6J_6, \ \ U_2=\frac{3}{2}J_1+3J_2, \ \ U_3=J_1+2J_2, \ \  U_4=\frac{1}{2}J_1+3J_2,$$
$$ U_5=6J_2, \ \ U_6=\frac{1}{2}J_1, \ \ U_7=3J_2,  \ \ U_{8}=J_2,$$
$$ U_9=J_1+J_2, \ \ U_{10}=\frac{1}{2}J_1+J_2, \ \  U_{11}=2J_2, \ \ U_{12}=0.$$

\begin{definition}
A configuration $\varphi$ is called a {\it ground state} of the relative Hamiltonian $H$ if $U(\varphi_b)=\min \{U_1, U_2,...,U_{12}\}$ for any $b\in M$.
\end{definition}

We set $C_i=\{\sigma_b:U(\sigma_b)=U_i\}$ and $U_i(J)=U(\sigma_b, J)$ if $\sigma_b\in C_i, i=1,2,...,12.$

If a ground state is a periodic (weakly periodic, translation invariant) configuration then we call it a {\it periodic (weakly periodic, translation invariant) ground state.}

Let $$A\subset \{1,2,...,k+1\}, \ \ H_{A}=\{x\in G_k: \sum_{j\in
A}w_j(x)\mbox{ is even}\},$$

 $$G_k^{(2)}=\{x\in G_k: |x|\mbox{ is even}\},
G_k^{(4)}=H_{A}\cap G_k^{(2)},$$ where $w_j(x)$ is the number of occurrences of
$a_j$ in $x$ and $|x|$ is the length of $x$, i.e.
$|x|=\sum_{j=1}^{k+1}w_j(x)$. Notice that $G_k^{(4)}$ is a normal subgroup of index 4 of  $G_k$.

Then we have $$G_k^{(4)}=\{x\in G_k:
|x|\mbox{ is even},\sum_{j\in A}w_j(x)\mbox{ is even}\}.$$ If $A=\{1,2,...,k+1\}$ then the normal subgroup $H_{A}$ coincides with the group $G_k^{(2)}$.

For any $i=1,2,...,12$ we put
\begin{equation}\label{u}
A_i=\{J\in\mathbf{R}^2: U_i=\min\{U_1, U_2,..., U_{12}\}\}.
\end{equation}

Quite cumbersome but not difficult calculations show that
$$
A_{1}=\{J\in \mathbf{R}^2: J_1\leq 0, J_2\leq 0\}\cup\{J\in \mathbf{R}^2: J_1\leq -6J_2, J_2\geq 0\},
$$
$$
A_{2}=\{J\in \mathbf{R}^2: J_1\geq0, -6J_2\leq J_1\leq-4J_2\}, $$$$ A_{3}=A_4=A_{10}=\{J\in \mathbf{R}^2: J_1=0, J_2=0\},
$$
$$
A_{5}=\{J\in \mathbf{R}^2: J_1\geq 0, J_2 \leq 0\},$$$$ A_{6}=\{J\in \mathbf{R}^2: J_2\geq0, -2J_2\leq J_1\leq 0\},
$$
$$
A_{7}=A_8=A_{11}=\{J\in \mathbf{R}^2: J_1\geq 0, J_2=0\},$$$$ A_{9}=\{J\in \mathbf{R}^2: J_2\leq 0, -4J_2\leq J_1\leq
-2J_2\},
$$
$$ A_{12}=\{J\in \mathbf{R}^2: J_1\leq 0, J_2\leq 0\},\quad \mbox{and} \quad \mathbf{R}^2 =\bigcup
\limits_{n} A_{n}.$$

\begin{theorem}\label{thm1}
For any class $C_i, i=1,2,...12,$ and any bounded configuration $\sigma_b\in C_i$, there exists a \textit{periodic configuration} $\varphi$ (on the Cayley tree) such that $\varphi_{b'}\in C_i$ for any $b'\in M$ and $\varphi_b=\sigma_b.$
\end{theorem}
\begin{proof}

For an arbitrary class $C_i, i=1,2,...,12,$ and $\sigma_b\in C_i$, we construct the configuration $\varphi$ as follows: without loss of generality, we can take the ball centered at $e\in G_3$ (where $e$ is the unit element of $G_3$) for the ball $b$, i.e., $b=\{e, a_1, a_2, a_3, a_4\}$.

We consider several cases.

{\it Case $C_1$}. In this case, we have $\sigma (x)=i$, $i \in \Phi,$ for any $x\in b$. The configuration $\varphi $ hence coincides with the translation-invariant configuration, i.e. $\varphi^{(i)}=\{\varphi^{(i)}(x)=i\}$, where $i\in \Phi.$

{\it Case $C_2$}. This case is considered in \cite{rras}. Let $H_{\{i\}}$ be normal subgroup of index two, and $G_k/H_{\{i\}}=\{H^{(2)}_0, H^{(2)}_1\}$ is quotient group, for any $i\in \{1,2,3,4\}$.
For any $l, m\in \Phi, l\neq m$, considering the following $H_{\{i\}}-$periodic configuration:
$$
\varphi_2^{(lm)}(x)=\left\{%
\begin{array}{ll}
    l, & \textrm{if} \ \ {x \in H^{(2)}_0},\\
    m, & \textrm{if} \ \ {x \in H^{(2)}_1}. \\
\end{array}%
\right. $$
In \cite{rras} it was proved, that $(\varphi_2^{(lm)})_{b'}\in C_2$ for any $b'\in M$.

{\it Case $C_3$}. Let $H_{\{i,j\}},$ $i,j\in \{1,2,3,4\}$ and $i\neq j$ be a normal subgroup of index two, and $G_k/H_{\{i,j\}}=\{H^{(3)}_0, H^{(3)}_1\}$ be the quotient group.
For any $l, m\in \Phi, l\neq m$, consider the following $H_{\{i,j\}}-$ periodic configuration:
$$
\varphi_3^{(lm)}(x)=\left\{%
\begin{array}{ll}
    l, & \textrm{if} \ \ {x \in H^{(3)}_0},\\
    m, & \textrm{if} \ \ {x \in H^{(3)}_1}. \\
\end{array}%
\right. $$

So we thus obtain a periodic configuration $\varphi_3^{(lm)}$ with period $p=2$ (equal to the index of the subgroup); then by construction  $(\varphi_3^{(lm)})_b=\sigma_b$. Now we shall prove that all restrictions  $(\varphi_3^{(lm)})_{b'}$ for any $b'\in M$ of the configuration $\varphi_3^{(lm)}$ belong to $C_3$.

Let $q_j(x)=|S_1(x)\cap H^{(3)}_j|, j=0,1$; where $S_1(x)=\{y\in G_k:\langle x, y\rangle \},$ the set of all nearest neighbors of $x\in G_k.$ Denote $Q(x)=(q_0(x), q_1(x)).$ Clearly $q_0(x)$ (resp. $q_1(x)$) is the number of points $y$ in $S_1(x)$ such that $\varphi_3^{(lm)}(y)=l$ (resp. $\varphi_3^{(lm)}(y)=m$). We note (see Chapter 1 of \cite{r}) that for every $x\in G_k$ there is a permutation $\pi_x$ of the coordinates of the vector $Q(e)$ (where $e$ as before is the identity of $G_k$) such that
$$
\pi_x Q(e)=Q(x).
$$

Moreover $Q(x)=Q(e)$ if $x\in H_0^{(3)}$ and $Q(x)=(q_1(e), q_0(e))$ if $x\in H_1^{(3)}.$ Thus for any $b'\in M$ we have $(i)$ if $c_{b'}\in H_0^{(3)}$ (where $c_{b'}$ is the center of $b'$) then $(\varphi_3^{(lm)})_{b'}\in C_3$, $(ii)$ if $c_{b'}\in H_1^{(3)},$ then $(\varphi_3^{(lm)})_{b'}\in C_3.$

{\it Case $C_4$}. Let $H_{\{i,j,r\}}$ be a normal subgroup of index two, and $G_k/H_{\{i,j,r\}}=\{H^{(4)}_0, H^{(4)}_1\}$ is the quotient group, for any $i,j,r\in \{1,2,3,4\}$ and $i\neq j, i\neq r, j\neq r$.
For any $l, m\in \Phi,  l\neq m$, considering the following $H_{\{i,j,r\}}-$ periodic configuration:
$$
\varphi_4^{(lm)}(x)=\left\{%
\begin{array}{ll}
    l, & \textrm{if} \ \ {x \in H^{(4)}_0},\\
    m, & \textrm{if} \ \ {x \in H^{(4)}_1}. \\
\end{array}%
\right. $$

We thus obtain a periodic configuration $\varphi_4^{(lm)}$ with period $p=2$; it is clear that $(\varphi_4^{(lm)})_{b'}\in C_4$ for any $b'\in M.$

{\it Case $C_5$}. Let $G_k/G_k^{(2)}=\{H^{(5)}_0, H^{(5)}_1\}$ is quotient group. For any $l, m\in \Phi, l\neq m$, consider the following $G_k^{(2)}-$ periodic configuration:
$$
\varphi_5^{(lm)}(x)=\left\{%
\begin{array}{ll}
    l, & \textrm{if} \ \ {x \in H^{(5)}_0},\\
    m, & \textrm{if} \ \ {x \in H^{(5)}_1}. \\
\end{array}%
\right. $$

It is easy to see (see \cite{rras}) that for each $b'\in M$ we have $(\varphi_5^{(lm)})_{b'}\in C_5$.

{\it Case $C_6$}. Let $G_3^{(6)}=H_{i}\cap H_{j}\cap H_{r},$ for any $i,j,r\in\{1,2,3,4\}, i\neq j, i\neq r, j\neq r.$ We note (see \cite{r}) that $G_3^{(6)}$ is a normal index-eight subgroup in $G_3$, and $G_3/G_3^{(6)}=\{H^{(6)}_0, H^{(6)}_1,..., H^{(6)}_7 \}$ is quotient group, where\\
$H^{(6)}_0=G_3^{(8)}=\{x\in G_3: w_i(x) $ is even$, w_j(x) $ is even$, w_r(x) $ is even$\}$,\\
$H^{(6)}_1=\{x\in G_3: w_i(x) $ is even$, w_j(x) $ is even$, w_r(x) $ is odd$\}$,\\
$H^{(6)}_2=\{x\in G_3: w_i(x) $ is even$, w_j(x) $ is odd$, w_r(x) $ is even$\}$,\\
$H^{(6)}_3=\{x\in G_3: w_i(x) $ is even$, w_j(x) $ is odd$, w_r(x) $ is odd$\}$,\\
$H^{(6)}_4=\{x\in G_3: w_i(x) $ is odd$, w_j(x) $ is even$, w_r(x) $ is even$\}$,\\
$H^{(6)}_5=\{x\in G_3: w_i(x) $ is odd$, w_j(x) $ is even$, w_r(x) $ is odd$\}$,\\
$H^{(6)}_6=\{x\in G_3: w_i(x) $ is odd$, w_j(x) $ is odd$, w_r(x) $ is even$\}$,\\
$H^{(6)}_7=\{x\in G_3: w_i(x) $ is odd$, w_j(x) $ is odd$, w_r(x) $ is odd$\}$.\\

For given $\sigma_b$, we have
$$
\sigma_b(x)=\left\{%
\begin{array}{ll}
    l, & \textrm{if} \ \ {x \in H^{(6)}_0}\cup { H^{(6)}_7}\cap b,\\
    m, & \textrm{if} \ \ {x \in H^{(6)}_1}\cup {H^{(6)}_6}\cap b,\\
    n, & \textrm{if} \ \ {x \in H^{(6)}_2}\cup {H^{(6)}_5}\cap b,\\
    p, & \textrm{if} \ \ {x \in H^{(6)}_3}\cup {H^{(6)}_4}\cap b.\\
\end{array}%
\right. $$
 For any $l,m,n,p\in \Phi, l\neq m, l\neq n, l\neq p, m\neq n, m\neq p, n\neq p$, continue the bounded configuration $\sigma_b\in C_6$ to the entire lattice $\Gamma^3$ (which is denoted by $\varphi_6^{(lmnp)}$, (see Fig.1)) as

%
\begin{figure}[b]
\sidecaption
\includegraphics[scale=.75]{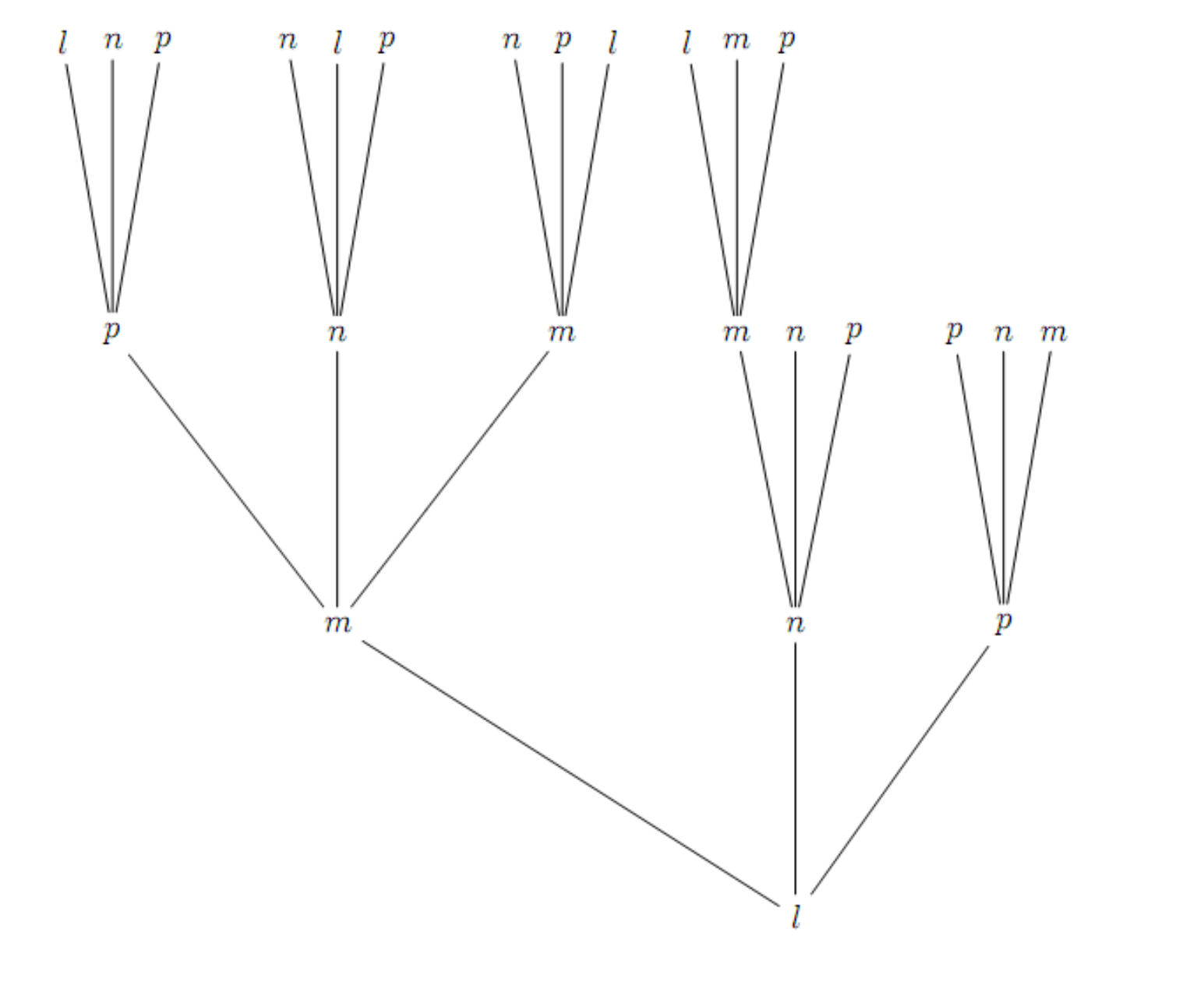}
%
%
\caption{Representation of the $G_k^{(8)}$ periodic configuration $\varphi_6^{(lmnp)}(x)$ on the Cayley tree of order $k=3$.}
\label{fig:1}       
\end{figure}

$$
\varphi_6^{(lmnp)}(x)=\left\{%
\begin{array}{ll}
    l, & \textrm{if} \ \ {x \in H^{(6)}_0\cup  H^{(6)}_7},\\
    m, & \textrm{if} \ \ {x \in H^{(6)}_1\cup  H^{(6)}_6},\\
    n, & \textrm{if} \ \ {x \in H^{(6)}_2\cup  H^{(6)}_5},\\
    p, & \textrm{if} \ \ {x \in H^{(6)}_3\cup  H^{(6)}_4}.\\
\end{array}%
\right. $$

{\it Case $C_7$}. Let $G_k/G_k^{(4)}=\{H^{(7)}_0,H^{(7)}_1,H^{(7)}_2,H^{(7)}_3\}$ be quotient group. In \cite{rras} it was proved that for $G_k^{(4)}-$ periodic configurations
$$
\varphi_7^{(lmn)}(x)=\left\{%
\begin{array}{ll}
    l, & \textrm{if} \ \ {x \in {H^{(7)}_0\cap H^{(7)}_1}},\\
    m, & \textrm{if} \ \ {x \in H^{(7)}_2},\\
    n, & \textrm{if} \ \ {x \in H^{(7)}_3},\\
\end{array}%
\right. $$
and
$$
\psi_7^{(lmn)}(x)=\left\{%
\begin{array}{ll}
    l, & \textrm{if} \ \ {x \in H^{(7)}_0},\\
    m, & \textrm{if} \ \ {x \in H^{(7)}_1},\\
    n, & \textrm{if} \ \ {x \in {H^{(7)}_2\cap H^{(7)}_3}},\\
\end{array}%
\right. $$
one has: $(\varphi_7^{(lmn)})_{b'}\in C_7$, $(\psi_7^{(lmn)})_{b'}\in C_7$ for all $l,m,n\in \Phi, l\neq m, l\neq n, m\neq n$ and for any $b'\in M$.

In \cite{rm2013} it was proved that $H_{\{1, 2, 3\}}-$ weakly periodic configurations
$$ \xi_7^{(lmn)}{(x)}=\left\{%
\begin{array}{ll}
    {l}, & \textrm{if} \ \ {x_{\downarrow} \in H_0}, \ x \in H_0 \\
    {m}, & \textrm{if} \ \ {x_{\downarrow} \in H_0}, \ x \in H_1 \\
    {n}, & \textrm{if} \ \ {x_{\downarrow} \in H_1}, \ x \in H_0 \\
    {l}, & \textrm{if} \ \ {x_{\downarrow} \in H_1}, \ x \in H_1, \\
\end{array}%
\right.$$
satisfy the following: $(\xi_7^{(lmn)})_{b'}\in C_7$ for all $l,m,n\in \Phi, l\neq m, l\neq n, m\neq n$ and for any $b'\in M$.

{\it Case $C_{8}$}. Let $G_3^{(8)}=H_{\{i, j\}}\cap H_{\{k\}}\cap H_{\{r\}},$ $i,j,k,r\in\{1,2,3,4\}, i\neq j, i\neq k, i\neq r, j\neq k, j\neq r, k\neq r.$ We note (see \cite{r}) that $G_3^{(8)}$ is a normal index-eight subgroup in $G_3$, and $G_3/G_3^{(8)}=\{H^{(8)}_0, H^{(8)}_1,..., H^{(8)}_7 \}$ is quotient group, where
$H^{(8)}_0=G_3^{(8)}=\{x\in G_3: w_i(x)+w_j(x) $ is even$, w_k(x) $ is even$, w_r(x) $ is even$\}$,\\
$H^{(8)}_1=\{x\in G_3: w_i(x)+w_j(x) $ is even$, w_k(x) $ is even$, w_r(x) $ is odd$\}$,\\
$H^{(8)}_2=\{x\in G_3: w_i(x)+w_j(x) $ is even$, w_k(x) $ is odd$, w_r(x) $ is even$\}$,\\
$H^{(8)}_3=\{x\in G_3: w_i(x)+w_j(x) $ is even$, w_k(x) $ is odd$, w_r(x) $ is odd$\}$,\\
$H^{(8)}_4=\{x\in G_3: w_i(x)+w_j(x) $ is odd$, w_k(x) $ is even$, w_r(x) $ is even$\}$,\\
$H^{(8)}_5=\{x\in G_3: w_i(x)+w_j(x) $ is odd$, w_k(x) $ is even$, w_r(x) $ is odd$\}$,\\
$H^{(8)}_6=\{x\in G_3: w_i(x)+w_j(x) $ is odd$, w_k(x) $ is odd$, w_r(x) $ is even$\}$,\\
$H^{(8)}_7=\{x\in G_3: w_i(x)+w_j(x) $ is odd$, w_k(x) $ is odd$, w_r(x) $ is odd$\}$.\\

In this case, for any $l,m,n,p\in \Phi, l\neq m, l\neq n, l\neq p, m\neq n, m\neq p, n\neq p$ we define the configuration $\varphi_{8}^{(lmnp)}$ as
$$
\varphi_{8}^{(lmnp)}(x)=\left\{%
\begin{array}{ll}
    l, & \textrm{if} \ \ {x \in H^{(8)}_0\cup  H^{(8)}_7},\\
    m, & \textrm{if} \ \ {x \in H^{(8)}_1\cup  H^{(8)}_6},\\
    n, & \textrm{if} \ \ {x \in H^{(8)}_2\cup  H^{(8)}_5},\\
    p, & \textrm{if} \ \ {x \in H^{(8)}_3\cup  H^{(8)}_4}.\\
\end{array}%
\right. $$

We thus obtain a periodic configuration $\varphi_{8}^{(lmnp)}$ with the period $p=8$ such that $(\varphi_{8}^{(lmnp)})_b=\sigma_b$, $(\varphi_{8}^{(lmnp)})_{b'}\in C_8$ for any $b'\in M.$

{\it Case $C_{9}$}. We consider a normal subgroup $\mathcal{H}_0\in G_3$ (see \cite{r}) of infinite index constructed as follows. Let the mapping $\pi_0:\{a_1, a_2, a_3, a_4\}\rightarrow \{e, a_1, a_2\}$ be defined by
$$
\pi_0(a_i)=\left\{%
\begin{array}{ll}
    a_i, & \textrm{if} \ \ {i=1, 2}\\
    e, & \textrm{if} \ \ {i\neq 1, 2}.\\
\end{array}%
\right. $$

Consider $$
f_0(x)=f_0(a_{i_1}a_{i_2}...a_{i_m})=\pi_0(a_{i_1})\pi_0(a_{i_2})...\pi_0(a_{i_m}).
$$

Then it is easy to see that $f_0$ is a homomorphism and hence $\mathcal{H}_0=\{x\in G_3: f_0(x)=e\}$ is a normal subgroup of infinity index.

Now we consider the factor group
$$
G_3/\mathcal{H}_0=\{\mathcal{H}_0, \mathcal{H}_0(a_1), \mathcal{H}_0(a_2), \mathcal{H}_0(a_1a_2), \mathcal{H}_0(a_2a_1),...\},
$$
where $\mathcal{H}_0(y)=\{x\in G_3:f_0(x)=y\}.$ We introduce the notations $$\mathcal{H}_n=\mathcal{H}_0(\underbrace{a_1a_2...}_n),  \mathcal{H}_{-n}=\mathcal{H}_0(\underbrace{a_2a_1...}_n).$$ In this notation, the factor group can be represented as
$$G_3/\mathcal{H}_0=\{... ,\mathcal{H}_{-2}, \mathcal{H}_{-2}, \mathcal{H}_0, \mathcal{H}_1, \mathcal{H}_2,...\}.$$

 It is known (see \cite{ri}), that for $x\in \mathcal{H}_{n}$ we have $|S_1(x)\cap \mathcal{H}_{n-1}|=1$, $|S_1(x)\cap \mathcal{H}_{n}|=k-1$, $|S_1(x)\cap \mathcal{H}_{n+1}|=1$.

Consider the following configuration
$$
\varphi^{(lm)}_{9}(x)=\left\{%
\begin{array}{ll}
    2nl, & \textrm{if} \ \ {x \in \mathcal{H}_n, n\neq 0},\\
    0, & \textrm{if} \ \ {x \in \mathcal{H}_0},\\
    (2n-1)m, & \textrm{if} \ \ {x \in \mathcal{H}_{-n}, n\neq 0},\\
\end{array}%
\right. $$
where $l, m\in \Phi, l\neq m$, $n=1,2,3....$

We thus obtain a periodic configuration $\varphi^{(lm)}_{9}$ with the infinity period, such that $(\varphi^{(lm)}_{9})_b=\sigma_b\in C_9$, and $(\varphi^{(lm)}_{9})_{b'}\in C_9$ for any $b'\in M$.

{\it Case $C_{10}$}. Let $S_3$ be the group of third-order permutations. We choose $\pi_0, \pi_1, \pi_2 \in S_3$ as
\begin{equation}
\pi_0=\left(
        \begin{array}{ccc}
          1 & 2 & 3 \\
          1 & 2 & 3 \\
        \end{array}
      \right),
\pi_1=\left(
        \begin{array}{ccc}
          1 & 2 & 3 \\
          1 & 3 & 2 \\
        \end{array}
      \right),
\pi_2=\left(
        \begin{array}{ccc}
          1 & 2 & 3 \\
          3 & 2 & 1 \\
        \end{array}
      \right).
\end{equation}
It is easily seen that $\pi_0=\pi_1^2=\pi_2^2$.

We consider the map $u:\{a_1,a_2,a_3,a_4\} \rightarrow \{\pi_1,\pi_2\}$
\begin{equation}
u(a_i)=\left\{
         \begin{array}{ll}
           \pi_1, & i=1, 2; \\
           \pi_2, & i=3, 4
         \end{array}
       \right.
\end{equation}
and assume that the function $f:G_3 \rightarrow S_3$ is defined as $$f(x)=f(a_{i_1}a_{i_2}\dots a_{i_n})=u(a_{i_1})\dots u(a_{i_n}).$$
Let
$$
\pi_3=\left(
        \begin{array}{ccc}
          1 & 2 & 3 \\
          3 & 1 & 2 \\
        \end{array}
      \right),
\pi_4=\left(
        \begin{array}{ccc}
          1 & 2 & 3 \\
          2 & 3 & 1 \\
        \end{array}
      \right),
\pi_5=\left(
        \begin{array}{ccc}
          1 & 2 & 3 \\
          2 & 1 & 3 \\
        \end{array}
      \right).
$$

We note (see \cite{gr}) that $H_{10}=\{x \in G_3: f(x)=\pi_0\}$ is a normal index-six subgroup. Let $G_3/H_{10}=\{\aleph_0, \dots, \aleph_5\}$ be the quotient group, where $$\aleph_i=\{x \in G_3: f(x)=\pi_i\}, i=\overline{0,5}.$$ In this case, we define the configuration $$\varphi_{10}^{(l,m,n)}(x)=\left\{
                                                                         \begin{array}{ll}
                                                                           l, & x \in \aleph_0 \cup \aleph_5, \\
                                                                           m, &  x \in \aleph_1 \cup \aleph_4, \\
                                                                           n, &  x \in \aleph_2 \cup \aleph_3,
                                                                         \end{array}
                                                                       \right.
$$
where $l,m,n\in \Phi, l\neq m, l\neq n, m\neq n.$

We thus obtain a periodic configuration $\varphi_{10}^{(l,m,n)}$ with the period six, such that $(\varphi_{10}^{(l,m,n)})_b=\sigma_b\in C_{10}$, $(\varphi_{10}^{(l,m,n)})_{b'}=\in C_{10}$ for any $b'\in M.$

{\it Case $C_{11}$}. Let $S_3$ be the group of third-order permutations.
It is easily seen that $\pi_0=\pi_1^2=\pi_5^2$.

We consider the map $u:\{a_1,a_2,a_3,a_4\} \rightarrow \{\pi_1,\pi_5\}$
\begin{equation}
u(a_i)=\left\{
         \begin{array}{ll}
           \pi_5, & i=1, 2; \\
           \pi_1, & i=3, 4,
         \end{array}
       \right.
\end{equation}
and assume that the function $f:G_3 \rightarrow S_3$ is defined as $$f(x)=f(a_{i_1}a_{i_2}\dots a_{i_n})=u(a_{i_1})\dots u(a_{i_n}).$$

We note (see \cite{gr}) that $H_{11}=\{x \in G_3: f(x)=\pi_0\}$ is a normal index-six subgroup. Let $G_3/H_{11}=\{\aleph_0, \dots, \aleph_5\}$ be the quotient group, where $$\aleph_i=\{x \in G_3: f(x)=\pi_i\}, i=\overline{0,5}.$$ In this case, we define the configuration $$\varphi_{11}^{(l,m,n)}(x)=\left\{
                                                                         \begin{array}{ll}
                                                                           l, & x \in \aleph_0 \cup \aleph_2,\\
                                                                           m, &  x \in \aleph_4 \cup \aleph_5, \\
                                                                           n, &  x \in \aleph_1 \cup \aleph_3,
                                                                         \end{array}
                                                                       \right.
$$
where $l,m,n\in \Phi, l\neq m, l\neq n, m\neq n.$

We thus obtain a periodic configuration $\varphi_{11}^{(l,m,n)}$ with the period six, such that $(\varphi_{11}^{(l,m,n)})_b=\sigma_b\in C_{11},$  $(\varphi_{11}^{(l,m,n)})_{b'}\in C_{11}$ for any $b'\in M.$

{\it Case $C_{12}$}.  Let $\mathcal{U}=\{(a_1a_2)^n\in G_3: n\in \mathbf{Z} \}.$ It is easy to see, that $\mathcal{U}$ is subgroup of the group $G_3$. Consider the set of right cosets $G_3/\mathcal{U}=\{\mathcal{U}, \mathcal{U}a_1,...,\mathcal{U}a_{k+1}, \mathcal{U}a_1a_2,...\}$ of $\mathcal{U}$ in $G_3$. We introduce the notations
$$
H_0=\mathcal{U}, H_1=\mathcal{U}a_1, ... ,H_{k+1}=\mathcal{U}a_{k+1}, H_{k+2}=\mathcal{U}a_1a_2,... .
$$

In this notation, the set of right coset can be represented as
$$
G_3/\mathcal{U}=\{H_0, H_1,...H_{k+1}, H_{k+2},...\}.
$$

Consider the following configuration: $\varphi^l_{12}(x)=l+i, $ if $x\in H_i$ for all $i=0,1,2,...$ and for any $l\in\Phi$. \\ Let $x\in H_n$, then $\varphi^l_{12}(x)=l+n$ and if $H_n=\mathcal{U}a_{j_1}a_{j_2}...a_{j_n}$, then for all $y\in S_1(x)$ we have $y\in \mathcal{U}a_{j_1}a_{j_2}...a_{j_n}a_t$, $t=1,2,3,4.$ By construction of configuration we have $\varphi^l_{12}(y)\neq \varphi^l_{12}(x)$ and $\varphi^l_{12}(y_1)\neq \varphi^l_{12}(y_2)$ for all $y, y_1, y_2\in S_1(x)$, $y_1\neq y_2.$

We thus obtain a $\mathcal{U}-$periodic configuration $\varphi^l_{12}$ with the infinity period, such that $(\varphi^l_{12})_{b'}\in C_{12}$ for any $b' \in M.$

\end{proof}

We set $B=A_1\cap A_2$, $B_0=A_1\cap A_5$, $B_1=A_2\cap A_9$, $B_2=A_9\cap A_6$, $B_3=A_6\cap A_{12}$, $\widetilde{A_1}=A_1\setminus (B\cup B_0)$, $\widetilde{A_2}=A_2\setminus (B_0\cup B_1)$, $\widetilde{A_5}=A_5\setminus (B_0\cup A_7)$, $\widetilde{A_6}=A_6\setminus (B_2\cup B_3)$, $\widetilde{A_9}=A_9\setminus (B_1\cup B_2)$ and $\widetilde{A_{12}}=A_{12}\setminus (B_3\cup A_7).$ Let $GS(H)$ be the set of all ground states, and let $GS_p(H)$ be the set of all periodic ground states.

\begin{theorem}\label{thm2}
A. If $J=(0,0)$, then $GS(H)=\Omega.$

B. 1. If $J\in \widetilde{A_1}$, then $GS_p(H)=\{\varphi^{(i)}: i\in \Phi\}.$\\
2. If $J\in \widetilde{A_2}$, then $GS_p(H)=\{\varphi_2^{(lm)}: l, m\in \Phi, l\neq m\}.$\\
3. If $J\in \widetilde{A_5}$, then $GS_p(H)=\{\varphi_5^{(lm)}: l, m\in \Phi, l\neq m\}.$\\
4. If $J\in \widetilde{A_6}$, then $GS_p(H)=\{\varphi_6^{(lmnp)}: l,m,n,p\in \Phi, l\neq m, l\neq n,l\neq p, m\neq n, \\ m\neq p, n\neq p\}.$\\
5. If $J\in \widetilde{A_9}$, then $GS_p(H)=\{\varphi_9^{(lm)}: l,m\in \Phi, l\neq m\}.$\\
6. If $J\in \widetilde{A_{12}}$, then $GS_p(H)=\{\varphi_{12}^{l}: l\in \Phi\}.$

C. 1. If $J\in B\setminus \{(0,0)\}$, then $GS_p(H)=\{\varphi^{(i)}, \varphi_2^{(lm)}: i, l, m\in \Phi, l\neq m\}$.\\
2. If $J\in B_0\setminus \{(0,0)\}$, then $GS_p(H)=\{\varphi^{(i)}, \varphi_5^{(lm)}: i, l, m\in \Phi, l\neq m\}$.\\
3. If $J\in B_1\setminus \{(0,0)\}$, then $GS_p(H)=\{\varphi_2^{(lm)}, \varphi_9^{(lm)}: i, l, m\in \Phi, l\neq m\}$.\\
4. If $J\in B_2\setminus \{(0,0)\}$ then $GS_p(H)=\{\varphi_6^{(lmnp)}, \varphi_9^{(lm)}: l,m,n,p\in \Phi, l\neq m, l\neq n,l\neq p, m\neq n,  m\neq p, n\neq p\}.$\\
5. If $J\in B_3\setminus \{(0,0)\}$ then $GS_p(H)=\{\varphi_6^{(lmnp)}, \varphi_{12}^{l}: l,m,n,p\in \Phi, l\neq m, l\neq n,l\neq p, m\neq n,  m\neq p, n\neq p\}.$\\
6. If $J\in A_8$, then periodic configuration $ \varphi_5^{(lm)}, \xi_7^{(lmn)}{(x)}, \psi_7^{(lmn)}(x), \varphi_{8}^{(lmnp)}(x), \varphi_{12}^{l}$ are periodic ground states, and weakly periodic configuration $\xi_7^{(lmn)}{(x)}$ is weakly periodic ground state, where $l,m,n,p\in \Phi, l\neq m, l\neq n, l\neq p, m\neq n, m\neq p, n\neq p$.
\end{theorem}

\begin{proof}
Case A is trivial. In cases B and C, for a given configuration $\sigma_b$ for which the energy $U(\sigma_b)$ is minimal, we can use Theorem \ref{thm1} to construct the periodic configurations.
\end{proof}

\begin{remark}
Since the set $\Phi$ is countable, it follows that the periodic and weakly periodic ground states described in Theorem \ref{thm2} are countable.
\end{remark}

\begin{acknowledgement}
The authors thank Professor U.A.Rozikov
for useful discussions.
\end{acknowledgement}


\begin{thebibliography}{99}

\bibitem{B}  Baxter, R.J.: Exactly Solved Models in Statistical Mechanics. Academic, London (1982).

\bibitem{g} Gankhodjaev N.N. Group represantation and automorphisms of the Cayley tree //Dokl.Akad.nauk Resp. Uzbekistan, (1994) No.4, pp.3-6.

\bibitem{gr} Ganikhodzhaev N.N., Rozikov U.A. Description of peridoic extreme Gibbs measures of some littice models on the Cayley tree//Theor. Math. Phys.,(1997), 111, 480Ц486.

\bibitem{rm} Mukhamedov F., Rozikov U., Mendes F.F. On contour arguments for the three state Potts model with
competing interactions on a semi-infinite Cayley tree//Journal of Mathematical Physics 48, 013301 (2007); https://doi.org/10.1063/1.2408398

\bibitem{r} Rozikov U.A.: Gibbs measures on Cayley trees.
World Scientific, (2013), ISBN-13: 978-9814513371  ISBN-10: 9814513377.

\bibitem{r2006} Rozikov U.A., A Constructive Description of Ground States and Gibbs Measures
for Ising Model With Two-Step Interactions on Cayley Tree, Jour.Statist.Phys. 122, (2006), p.217Ц-235.

\bibitem{rm2010}	M.M.Rahmatullaev. Desription  of weak periodic ground states of Ising model with competing interactions on Cayley tree. // Applied Mathematics and Information Sciences.  2010, - 4(2). -p. 237-241.

\bibitem{br} Botirov G.I., Rozikov U.A. Potts model with competing interactions on the Cayley tree: The contour method// {\it Theor. Math. Phys}, {\bf 153}, (2007), pp. 1423-1433.

\bibitem{ri} Rozikov U.A., Ishankulov F.T. Description of periodic
$p$-harmonic functions on Cayley tree. {\it NoDEA. Nonlinear Differential Equations and Appl.} 2010, V.17, No. 2, p.153-160.

\bibitem{rr} Rozikov U.A., Rahmatullaev M.M. Weakly periodic ground states and Gibbs measures for the Ising model with cmpeting interactions on the Cayley tree//{\it Theor. Math. Phys}, {\bf 160}, (2009), pp. 1292-1300.

\bibitem{rras} Rahmatullaev M.M., Rasulova M.A. Periodic and weakly periodic ground states for the Potts model with competing interactions on the Cayley tree// {\it Siberian advances in Mathematics}, {\bf 26}, 3, (2016), pp. 215-229.

\bibitem{rm2013} Rahmatullaev M.M. Weakly periodic Gibbs measures and ground states for the Potts model with competing interactions on the Cayley tree//{\it Theor. Math. Phys}, {\bf 176}, (2013), pp. 1236-1251.


\bibitem{re} Rozikov, U.A. Eshkabilov, Yu.Kh.: On models with uncountable set
of spin values on a Cayley tree: Integral equations. {\it Math.
Phys. Anal. Geom.} {\bf 13} (2010), 275-286.


\end{thebibliography}
\end{document}